\documentclass[journal,draftcls,onecolumn,12pt,twoside]{IEEEtran}
\usepackage{amsmath, amsthm, amsfonts, amssymb}
\usepackage{graphicx}
\usepackage{subfigure}
\usepackage{multirow}
\usepackage{cite}
\usepackage{enumerate}

\newcommand{\seq
}[1]{\mathbf{#1}}

\newtheorem{theorem}{\bf Theorem}
\newtheorem{prop}[theorem]{\bf Proposition}
\newtheorem{corollary}[theorem]{\bf Corollary}
\newtheorem{lemma}[theorem]{\bf Lemma}

\newtheorem{definition}[theorem]{\bf Definition}

\long\def\symbolfootnote[#1]#2{\begingroup
\def\thefootnote{\fnsymbol{footnote}}\footnote[#1]{#2}\endgroup}

\begin{document}
\title{Protocol Sequences for Multiple-Packet Reception}

\author{Yijin Zhang, Yuan-Hsun Lo, Feng Shu and Wing Shing Wong
\thanks{
The material in this paper was presented in part at the IEEE International Symposium on Information Theory, Honolulu, HI, July 2014.
This work was partially supported by the National Natural Science Foundation of China (No. 61301107, 61174060, 61271230), the Hong Kong RGC Earmarked Grant
CUHK414012 and the Shenzhen Knowledge Innovation Program JCYJ20130401172046453.}
\thanks{Y. Zhang and F. Shu are with the School of Electronic and Optical Engineering, Nanjing University of Science and Technology, Nanjing, China.
E-mail: yijin.zhang@gmail.com; shufeng@njust.edu.cn.}
\thanks{Y.-H. Lo is with the Department of Mathematics, National Taiwan Normal University, Taipei, Taiwan. E-mail: yhlo0830@gmail.com.}
\thanks{W. S. Wong is with the Department of Information Engineering, The Chinese University of Hong Kong, Shatin, N.~T., Hong Kong.
E-mail: wswong@ie.cuhk.edu.hk.}
}
\maketitle

\begin{abstract}
Consider a time slotted communication channel shared by $K$ active users and a single receiver.
It is assumed that the receiver has the ability of the multiple-packet reception (MPR) to correctly receive at most $\gamma$ ($1 \leq \gamma < K$) simultaneously transmitted packets.
Each user accesses the channel following a specific periodical binary sequence, called the protocol sequence, and transmits a packet within a channel slot if and only if the sequence value is equal to one.
The fluctuation in throughput is incurred by inevitable random relative shifts among the users due to the lack of feedback.
A set of protocol sequences is said to be throughput-invariant (TI) if it can be employed to produce invariant throughput for any relative shifts, i.e., maximize the worst-case throughput.
It was shown in the literature that the TI property without considering MPR (i.e., $\gamma=1$) can be achieved by using shift-invariant (SI) sequences, whose generalized Hamming cross-correlation is independent of relative shifts.
This paper investigates TI sequences for MPR; results obtained include achievable throughput value, a lower bound on the sequence period, an optimal construction of TI sequences that achieves the lower bound on the sequence period, and intrinsic structure of TI sequences.
In addition, we present a practical packet decoding mechanism for TI sequences that incorporates packet header, forward error-correcting code, and advanced physical layer blind signal separation techniques.

\end{abstract}

\section{Introduction}
The idea of using {\em protocol sequences} to define deterministic multiaccess protocols for a collision channel without feedback was proposed by Massey and Mathys in~\cite{Massey85}, and recently has attracted many research revisits~\cite{Nguyen92,GyorfiVajda93,Wong07,CWS08,SCSW,CRT} for different design criteria and applications.
Compared with time division multiple access (TDMA), ALOHA and carrier sense multiple access (CSMA), protocol sequences do not require stringent synchronization, channel monitoring, backoff algorithm and packet retransmissions.
Such a simplicity is particularly desirable in ad hoc networks and sensor networks, in which well-coordinated transmissions and time synchronization may be difficult to achieve due to user mobility, time-varying propagation delays and energy constraints.
Moreover, in contrast to the random and contention based schemes, protocol sequences in~\cite{Nguyen92,GyorfiVajda93,Wong07,CWS08,SCSW,CRT} can respectively provide a positive short-term throughput guarantee with probability one in the worst case.

Previous studies on multiaccess protocols have traditionally assumed a collision channel model of {\em single-packet reception} (SPR), in which a packet is received correctly only if it is not involved in a collision, i.e., does not overlap with another.
However, the assumption of SPR becomes more and more unsuitable in practice, due to recent advances at reception techniques of the physical (PHY) layer, such as antenna arrays, CDMA technique and beamforming algorithms, which can be employed to ensure that the receiver has the ability of {\em multiple-packet reception} (MPR)~\cite{Ghez} to receive multiple packets simultaneously.
In this paper, we restrict our attention to protocol sequences for the $\gamma$-user MPR channel, in which a packet can be received error-free only if at most $\gamma-1$ other packets are being transmitted simultaneously.
We refer to $\gamma$ as the {\em MPR capability} which has been commonly assumed in~\cite{ZZL09,Babich10,AlohaMPR13,MPR_MAC13,Chan13} for studying ALOHA and CSMA.
However, it is expected that protocol sequences will also behave differently from what were reported in~\cite{Nguyen92,GyorfiVajda93,Wong07,CWS08,CRT,SCSW} for $\gamma=1$ (i.e., SPR).
Until only recently there has been no research published on protocol sequences for MPR.
Only one related result~\cite{Tinguely05} stated that there is no need of employing protocol sequences for a channel with a sufficiently large recovering probability of collided packets, which is different from the MPR capability discussed here.

The {\em effective throughput} of a user under the MPR capability $\gamma$ is defined as the fraction of packets it can be sent out without suffering any collision in which more than $\gamma$ users are involved.
Due to a lack of feedback information, the relative shifts among users are unknown to each other, and thus incur performance variation in throughput.
As argued in~\cite{CWS08,SCSW}, our main goal of designing protocol sequences is to maximize the worst-case individual throughput for any relative shifts, i.e., minimize the variation in throughput.
If the throughput of each user is constant and independent of any relative shifts, then the assigned set of protocol sequences is said to be {\em throughput-invariant} (TI).
It was shown in~\cite{Massey85,CWS08,SCSW} that this zero-variance on throughput for $\gamma=1$ can be achieved by using {\em shift-invariant} (SI) sequences, which enjoy the special property that their {\em generalized Hamming cross-correlation} is independent of relative shifts.
The generalized Hamming cross-correlation here is a generalization of {\em pairwise Hamming cross-correlation} and defined for all nonempty subsets of users.
However, the question of whether there exist TI sequences for $\gamma=1$ which are shorter than SI sequences has not been answered.
Moreover, the impact of MPR capability on the throughput performance and sequence design of TI sequences is not explored either.
As such, this paper is a first attempt to study TI sequences for MPR, and can be viewed as a generalization of results in~\cite{SCSW}.

The remainder of this paper is organized as follows.
After setting up the channel model and notation in Section II, we in Section III derive the throughput value of a system with MPR capability.
In Section IV, lower bounds on the period of TI sequences with MPR capability are presented.
In Section V, we use a known construction of SI sequences to design TI sequences for any MPR capability, which are optimal in the sense that the sequence period achieves the lower bound.
Section VI proposes a mechanism of identification and decoding for TI sequences.
It is further shown in Section VII that TI sequences must be SI in many specific cases, which indicates that the SI sequence set is the unique solution to the TI problem with MPR.
Finally, in Section VIII we close the discussion with some concluding remarks.

\section{Channel Model and Notation}
Consider a time slotted communication channel shared by $K$ active users and a single receiver.
It is assumed that each of these users always has a fixed-length packet to send, knows the slot boundaries and transmits its packet within a channel slot. However, these users do not know the relative shifts of other users.
Let $\tau_i$, an integer measured in unit of slot duration, denote the relative shift of user $i$ for $i=1,2,\ldots,K$.
Following~\cite{Massey85}, we define a deterministic binary sequence, $\seq{s}_i:= [ s_i(0)\ s_i(1)\ \ldots\ s_i(L-1)]$, called {\em protocol sequence} to schedule the packet transmission of user~$i$ for $i=1,2,\ldots,K$, where $L$ is the common sequence period of all $K$ sequences.
We consider that the channel is slot-synchronous so that there exists a system-wide slot labeling, $t$, and user $i$ transmits a packet at slot $t$ if $s_i(t+\tau_i) = 1$, and keeps silent if $s_i(t+\tau_i) = 0$, in which the addition by $\tau_i$ is in modulo $L$ arithmetic.

Following the assumption of the $\gamma$-user MPR channel, we focus on the non-trivial case in which $\gamma<K$.
A transmitted packet is correctly received if at most $\gamma-1$ other packets are being transmitted at the same slot, and is considered lost otherwise.
For practical considerations, the users can employ forward error-correcting code across packets to recover data lost, as explained in Section VI.

Some notation and definitions used in this paper are stated below.

\begin{definition}
For $i=1,\ldots,K$, define the \emph{duty factor} $f_i$ of a protocol sequence $\seq{s}_i$, as the fraction of ones in $\seq{s}_i$, namely,
\[
f_i:=\frac{1}{L}\sum_{t=0}^{L-1}s_i(t).
\]
The \emph{cyclic shift} of $\seq{s}_i$ by $\tau_i$ is defined as
\[
\seq{s}_i^{(\tau_i)}:= [ s_i(\tau_i)\ s_i(1+\tau_i)\ \ldots\ s_i(L-1+\tau_i) ],
\]
where the addition is taken modulo $L$.
Note that the $t$-th bit of $\seq{s}_i^{(\tau_i)}$ is denoted by $s_i(t+\tau_i)$.
\end{definition}

\begin{definition}
Let $b_j \in \{0,1\}$ for $j=1,\ldots,K$.
For a system with the MPR capability $\gamma$, the \emph{throughput} of user $i$ with the assigned sequence $\seq{s}_i$ for given relative shifts $\tau_1,\ldots,\tau_K$ is defined as
\begin{equation}
R_i(\tau_1,\tau_2,\ldots,\tau_K)=\frac{1}{L} \sum_{b_i=1, \atop q_i\leq \gamma-1} N(b_1,\ldots,b_K|\seq{s}_1^{(\tau_1)},\ldots,\seq{s}_K^{(\tau_K)}),
\label{R}
\end{equation}
in which $q_i=\sum_{j\neq i}b_j$ and $N(b_1,\ldots,b_K|\seq{s}_1^{(\tau_1)},\ldots,\seq{s}_K^{(\tau_K)})$ denotes the number of time indices $t$, $0 \leq t < L$, such that $s_j(t+\tau_j)= b_j$ for all $j$.
This computes the fraction of time slots in which at most $\gamma$ users including user $i$ are transmitting.
Note that the summation in~\eqref{R} is taken over $(b_1,\ldots,b_K)$ with $b_i=1, q_i\leq \gamma-1$.

A sequence set $\{\seq{s}_1,\seq{s}_1,\ldots,\seq{s}_K\}$ is said to be TI with the MPR capability $\gamma$ if $R_i(\tau_1,\tau_2,\ldots,\tau_K)$ is a constant function of $\tau_1, \tau_2,\ldots, \tau_K$ for any $i$.
For simplicity, we sometimes use $R_i$ to denote the throughput of user $i$.
To avoid the uninteresting cases, in this paper we only consider TI sequences that ensure $R_{i}$ is strictly bigger than zero for any $i$.
\end{definition}

\begin{definition}
We identify the $K$ users by means of the index set $$\mathcal{K} := \{1,2,\ldots, K\}.$$
Let $\mathcal{O}_K$ be the set
\[
 \bigcup_{n=1}^{K} \left\{ (i_1, \ldots, i_n) \in\mathcal{K}^n:\, i_1 <i_2< \cdots< i_n \right\}.
\]
An element in $\mathcal{O}_K$ corresponds to an ordered tuple of users.
For $\mathsf{A} = (i_1,i_2,\ldots,i_n) \in \mathcal{O}_K$, the \emph{generalized Hamming cross-correlation} associated with $\mathsf{A}$ is defined as
\[
 H(\tau_{i_1}, \ldots, \tau_{i_n}; \mathsf{A}) := \sum_{t=0}^{L-1} \prod_{j=1}^n s_{i_j}(t+\tau_{i_j}).
\]
If $|\mathsf{A}|=2$, the generalized Hamming cross-correlation is reduced to the {\em pairwise Hamming cross-correlation} function.

We further introduce the following definitions by means of the generalized Hamming cross-correlation:

(i) Given an ordered tuple $\mathsf{A} \in \mathcal{O}_K$, then $H(\tau_{i_1}, \ldots, \tau_{i_n}; \mathsf{A})$ is said to be SI if it is a constant for any $\tau_{i_1}, \ldots, \tau_{i_n}$.

(ii) A sequence set is said to be SI~\cite{SCSW} if $H(\tau_{i_1}, \ldots, \tau_{i_n}; \mathsf{A})$ is SI for every $\mathsf{A}$ in $\mathcal{O}_K$.

(iii) A sequence set is said to be {\em pairwise} SI~\cite{ZSW} if $H(\tau_{i_1}, \ldots, \tau_{i_n}; \mathsf{A})$ is SI for every $\mathsf{A}$ in $\mathcal{O}_K$ with $|\mathsf{A}|=2$.
\end{definition}

\section{Throughput of TI sequences}
Let $\seq{s}_1,\seq{s}_2,\ldots,\seq{s}_n$ be $n$ binary sequences with a common period $L$.
For $b_1,\ldots,b_n\in\{0,1\}$, Shum et al.~\cite{SCSW} showed that
\begin{equation}
\sum_{\tau_1=0}^{L-1}\ldots \sum_{\tau_n=0}^{L-1} N(b_1,\ldots,b_n|\seq{s}_1^{(\tau_1)},\ldots,\seq{s}_n^{(\tau_n)})=L\prod_{i=1}^nN(b_i|\seq{s}_i).
\label{eq:Shum09}
\end{equation}

The main result in this section is summarized in the following.

\begin{theorem}\label{thm:throughput}
Let $\seq{s}_1,\seq{s}_2,\ldots,\seq{s}_K$ be $K$ TI sequences with the MPR capability $\gamma$, $1 \leq \gamma < K$, and duty factors $f_1,f_2,\ldots,f_K$ respectively.
Then we have
\begin{equation}
R_i=f_i \sum_{\mathsf{H}\subseteq \mathcal{K} \setminus \{i\} \atop |\mathsf{H}|< \gamma} \prod_{j\in \mathsf{H}}f_j \prod_{k\in \mathcal{K} \setminus (\{i\} \cup \mathsf{H})}(1-f_k).
\end{equation}
\label{T}
\end{theorem}
\begin{proof}
Suppose that the relative shifts of $\seq{s}_i$ is $\tau_i$, for $i=1,2,\ldots,K$.
We can treat $\tau_1, \tau_2,\ldots ,\tau_K$ as independent and uniformly distributed random variables that are equally likely to take on any of $L$ values: $0,1,\ldots,L-1$.
After taking the expectation over $(\tau_1, \tau_2,\ldots ,\tau_K)$, we obtain the average throughput of user $i$ as the following:
\begin{align}
\mathbb{E} &\Big ( R_i(\tau_1,\tau_2,\ldots,\tau_K) \Big )  \notag \\
&=\frac{1}{L}\mathbb{E}\Big ( \sum_{q_i\leq \gamma-1,b_i=1} N(b_1,\ldots,b_K|\seq{s}_1^{(\tau_1)},\ldots,\seq{s}_K^{(\tau_K)})\Big) \label{eq:N}\\
&=\frac{1}{L}\sum_{q_i\leq \gamma-1,b_i=1} \mathbb{E}\Big ( N(b_1,\ldots,b_K|\seq{s}_1^{(\tau_1)},\ldots,\seq{s}_K^{(\tau_K)})\Big) \notag \\
&=\frac{1}{L}\sum_{q_i\leq \gamma-1,b_i=1} \frac{1}{L^{K-1}}\prod_{j=1}^KN(b_j|\seq{s}_j) \label{eq:N2} \\
&=\frac{1}{L^{K-1}} \sum_{q_i\leq \gamma-1}\frac{N(1|\seq{s}_i)}{L}  \prod_{j=1,j \neq i}^KN(b_j|\seq{s}_j) \notag \\
&=\frac{1}{L^{K-1}}f_i \sum_{q_i\leq \gamma-1} \prod_{j=1,j \neq i}^KN(b_j|\seq{s}_j) \notag \\
&=f_i \sum_{\mathsf{H}\subseteq \mathcal{K} \setminus \{i\} \atop |\mathsf{H}|< \gamma} \prod_{j\in \mathsf{H}}f_j \prod_{k\in \mathcal{K} \setminus (\{i\} \cup \mathsf{H})}(1-f_k), \notag
\end{align}
where $q_i=\sum_{j\neq i} b_i$.
\eqref{eq:N} directly follows from~\eqref{R}, while~\eqref{eq:N2} is due to~\eqref{eq:Shum09}.
Furthermore, since $R_i(\tau_1,\ldots,\tau_K)$ is a constant function of $\tau_1, \ldots, \tau_K$, it must be equal to its average value.
This completes the proof.
\end{proof}

Note that Theorem~\ref{thm:throughput} is a generalization of \cite[Thm. 3]{SCSW}, which focuses on the case of $\gamma=1$.

For the symmetric case that each user has the same duty factor $f$, all users have the same throughput:
\begin{equation}
\sum_{j=0}^{\gamma-1} \binom{K-1}{j} f^{j+1} (1-f)^{K-1-j}.
\label{eq:T}
\end{equation}
The system throughput in the symmetric case is plotted in Fig. 1 for $10 \leq K \leq 50$ with $\gamma=1,5,10$ and $f=1/10,1/20$, respectively.
\begin{figure}
\begin{center}
  \includegraphics[width=4in]{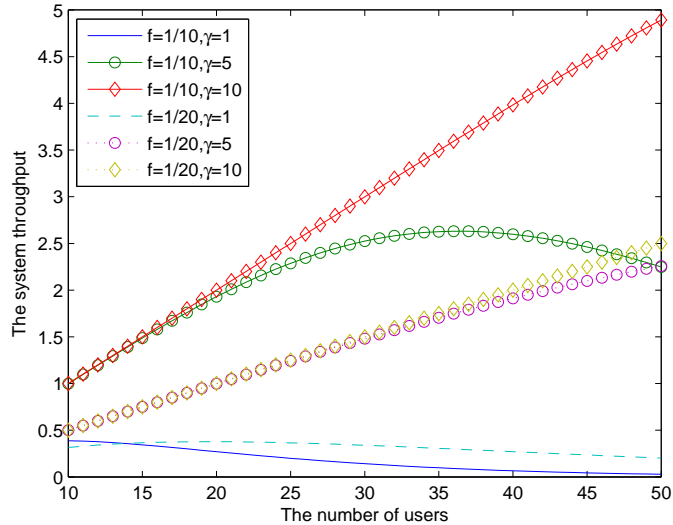}
\end{center}
\caption{The symmetric system throughput for $10 \leq K \leq 50$.}
\label{fig:throughput}
\end{figure}
We can see from \eqref{eq:T} that there is an optimal duty factor for maximizing the throughput of a given user number. For $\gamma=1$ the optimal value is $1/K$, but for the other cases the closed-form expression is difficult to obtain. See numerical results of optimal $f$ in Fig.~\ref{fig:optD} for $K=20$ and $\gamma=1,2,3,5,8,10,12,15$.

\begin{figure}
\begin{center}
  \includegraphics[width=4in]{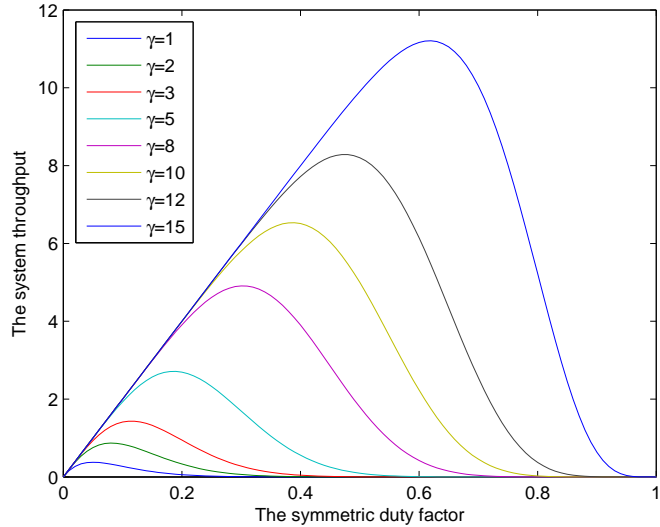}
\end{center}
\caption{The symmetric system throughput for $K=20$ and different $\gamma$.}
\label{fig:optD}
\end{figure}

\section{Lower bound on Minimum Period}
A long sequence period could result in large variation in throughput on a short-time scale.
With a weak assumption on duty factors, we derive lower bounds on the period of TI sequences for any $\gamma$ in this section.
These are clearly constraints on constructing TI sequences for small $L$ values.
Let $\gcd(x, y)$ denote the greatest common divisor of integers $x$ and $y$.

\begin{definition} \label{def:lambda_theta}
Consider $K$ binary sequences $\seq{s}_1, \seq{s}_2, \ldots, \seq{s}_K$ with a common period $L$.
Given an ordered tuple $\mathsf{A}=(i_1,\ldots,i_M)\in\mathcal{O}_K$ and relative shifts $\tau_{i_1}, \ldots, \tau_{i_M}$ for some $M\leq K$, let $\theta_j(\tau_{i_1}, \ldots, \tau_{i_M};\mathsf{A})$, for $j=0,1,\ldots,M$, denote the number of time indices $t$, $0\leq t<L$, such that there are exactly $j$ `1's among ${s}_{i_1}(t+\tau_{i_1}), \ldots, {s}_{i_M}(t+\tau_{i_M})$.
Then, define
\[
\theta_{\leq j}(\tau_{i_1}, \ldots, \tau_{i_M};\mathsf{A}) = \sum_{k=0}^j \theta_k(\tau_{i_1}, \ldots, \tau_{i_M};\mathsf{A}).
\]
Note that $\theta_M(\tau_{i_1}, \ldots, \tau_{i_M};\mathsf{A}) = H(\tau_{i_1}, \ldots, \tau_{i_M};\mathsf{A})$.
\end{definition}


\begin{prop}
Consider a set of $K$ sequences which is TI with the MPR capability $\gamma$ ($1 \leq \gamma <K$).
If $R_i>0$ for any $i$, then there are at most $\gamma-1$ all-one sequences.
\label{prop:allone}
\end{prop}
\begin{proof}
Suppose there are $\gamma$ all-one sequences.
Then all packets from other $K-\gamma$ sequences cannot be received error-free.
It implies $R_i=0$ for some $i$, which contradicts the assumption.
\end{proof}

\begin{theorem}
Let $\gamma$ be any integer with $1 \leq \gamma < K$.
If a set of $K$ binary sequences is TI with the MPR capability $\gamma$, then this sequence set is pairwise SI.
\label{pSI}
\end{theorem}
\begin{proof}
Denote by $\seq{s}_1,\ldots,\seq{s}_K$ the $K$ sequences.
The proof of the case $K=2$ is straightforward.
For $K>2$, we aim to show that $\theta_2(\tau_{i_1},\tau_{i_2};\mathsf{A})$ is a constant function of $\tau_{i_1},\tau_{i_2}$ for any $\mathsf{A}=(i_1,i_2)\in\mathcal{O}_K$.
Without loss of generality, let $\mathsf{A}=(1,2)$, and let $\mathsf{B}=(3,\ldots,K)$.
That is, we divide these $K$ sequences into two parts: $\mathcal{A}=\{\seq{s}_1,\seq{s}_2\}$ and $\mathcal{B}=\{\seq{s}_3,\ldots,\seq{s}_K\}$.

First, we fix the relative shifts $\tau^*_3,\ldots,\tau^*_K$ of sequences in $\mathcal{B}$ such that $\theta_{\gamma-1}(\tau^*_{3},\ldots,\tau^*_K;\mathsf{B})>0$.
This combination of relative shifts can always be found since there are at most $\gamma-1$ all-one sequences in $\mathcal{B}$ for the case $1 \leq \gamma < K$ following Proposition~\ref{prop:allone}.
Suppose in $\mathcal{B}$ that there are exactly $h$ all-one sequences.
As $h\leq \gamma-1 \leq K-2$, then we cyclically shift some $\gamma-h-1$ sequences which are not all-ones such that the first time slot of them are all equal to one, and cyclically shift the remaining $K-\gamma-1$ sequences such that the first time slot of them are all equal to zero.
Hence there are exactly $\gamma-1$ ones in the first time slot.

For $\tau_1,\tau_2$, we define
\[
T_1(\tau_1,\tau_2) := \sum_{b_1+b_2=1 \atop b_3+\ldots+b_K < \gamma}  N(b_1,\ldots,b_K|\seq{s}_1^{(\tau_1)}, \seq{s}_2^{(\tau_2)}, \seq{s}_3^{(\tau^*_3)},\ldots,\seq{s}_K^{(\tau^*_K)})
\]
and
\[
T_2(\tau_1,\tau_2) := \sum_{b_1+b_2=2 \atop b_3+\ldots+b_K < \gamma-1}  N(b_1,\ldots,b_K|\seq{s}_1^{(\tau_1)}, \seq{s}_2^{(\tau_2)}, \seq{s}_3^{(\tau^*_3)},\ldots,\seq{s}_K^{(\tau^*_K)}).
\]
We assume the relative shifts $\tau_1, \tau_2$ are uniformly distributed in $0,1,\ldots,L-1$.
After taking the expectation over $\tau_1, \tau_2$, we have
\begin{align}
\mathbb{E}( LR_1+LR_2) =& \mathbb{E}(T_1(\tau_1,\tau_2)+2T_2(\tau_1,\tau_2))=\mathbb{E}(T_1(\tau_1,\tau_2))+2\mathbb{E}(T_2(\tau_1,\tau_2)) \label{eq:average_1} \\
=&\frac{1}{L}\sum_{i=1}^2 i\theta_i(\tau_1,\tau_2;\mathsf{A}) \theta_{\leq \gamma-i}(\tau^*_{3},\ldots,\tau^*_K;\mathsf{B}). \label{eq:average_2}
\end{align}
The proof of (\ref{eq:average_2}) is relegated to Appendix~\ref{sec:app_1}.

We now change the pair of relative shifts from $(\tau_1,\tau_2)$ to $(\tau_1',\tau_2')$.
Let
\begin{equation}
\sigma := \theta_2(\tau_1',\tau_2';\mathsf{A}) - \theta_2(\tau_1,\tau_2;\mathsf{A}). \label{eq:delta_pair_2}
\end{equation}
By the fact that
$$\theta_1(\tau_1,\tau_2;\mathsf{A})+2\theta_2(\tau_1,\tau_2;\mathsf{A})=Lf_1+Lf_2=\theta_1(\tau_1',\tau_2';\mathsf{A})+2\theta_2(\tau_1',\tau_2';\mathsf{A}),$$
we have
\begin{equation}
\theta_1(\tau_1',\tau_2';\mathsf{A})-\theta_1(\tau_1,\tau_2;\mathsf{A}) = -2\sigma. \label{eq:delta_pair_1}
\end{equation}
Since $R_1+R_2$ has zero-variance, by \eqref{eq:average_1} and \eqref{eq:average_2}, we have
\begin{align*}
\sum_{i=1}^2&i \, \theta_i(\tau_1,\tau_2;\mathsf{A}) \theta_{\leq \gamma-i}(\tau^*_{3},\ldots,\tau^*_K;\mathsf{B}) \\
&= \sum_{i=1}^2i \, \theta_i(\tau_1',\tau_2';\mathsf{A}) \theta_{\leq \gamma-i}(\tau^*_{3},\ldots,\tau^*_K;\mathsf{B}).
\end{align*}
Then it follows from~\eqref{eq:delta_pair_2} and \eqref{eq:delta_pair_1} that
\[
\sigma \, \theta_{\gamma-1}(\tau^*_{3},\ldots,\tau^*_K;\mathsf{B}) = 0,
\]
which implies that $\sigma=0$ because of $\theta_{\gamma-1}(\tau^*_{3},\ldots,\tau^*_K;\mathsf{B})>0$ due to the choice of $\tau^*_{3},\ldots,\tau^*_K$.
Thus, $\theta_2(\tau_1,\tau_2;\mathsf{A})$ is a constant function of $\tau_1,\tau_2$.

Since $\theta_2(\tau_1,\tau_2;\mathsf{A}) = H(\tau_1,\tau_2;\mathsf{A})$ and the choice of $\mathsf{A}$ is arbitrary, we conclude that a set of these $K$ sequences is pairwise SI for any $\gamma < K$.
\end{proof}

Compared with the SI property which has been proved as a sufficient condition of a sequence set being TI for $\gamma=1$~\cite{SCSW}, pairwise SI is conceptually a weaker requirement on sequence design.
However, they are known to enjoy the same lower bound on the minimum sequence period for some special form of duty factors.
Given any set of $K$ pairwise SI sequences with duty factors $n_i/d_i$, where $\gcd(n_i,d_i)=1$ for all $i$, Zhang et al. \cite[Thm. 1]{ZSW} proved that its common period is divisible by $d_1d_2\cdots d_K$.
By Theorem~\ref{pSI}, we then have the following result.

\begin{corollary} \label{cor:period}
Let $\gamma$ be an integer with $1 \leq \gamma < K$.
If a set of $K$ binary sequences with duty factors $n_i/d_i$, where $\gcd(n_i,d_i) = 1$ for all $i$, is TI with the MPR capability $\gamma$, then its common period is divisible by $d_1 d_2 \cdots d_K$. In particular, the minimum common period is at least $d_1 d_2 \cdots d_K$.
\end{corollary}

With duty factors $n_i/d_i$ where $\gcd(n_i,d_i) = 1$ for all $i$, Corollary~\ref{cor:period} further obtains that the lower bound on the period of TI sequences still grows exponentially with $K$ for any $\gamma$,
although their combinatorial design requirement is different from that of pairwise SI and SI sequences.

\section{An Optimal Construction}

Shum et al.~\cite{SCSW} showed that any SI sequence set is TI for the classical model ($\gamma=1$).
In this section, we extend this property to general $\gamma$ by means of the following result.

\begin{theorem}[{\cite[Thm. 1]{SCSW}}]
The sequence set $\seq{s}_1,\seq{s}_2,\ldots,\seq{s}_K$ is SI if and only if for each choice of $b_1,\ldots,b_K$,
$N(b_1,\ldots,b_K|\seq{s}_1^{(\tau_1)},\ldots,\seq{s}_K^{(\tau_K)})$ is a constant function of $\tau_1,\ldots,\tau_K$.
\label{Con1}
\end{theorem}

\begin{theorem}
If a sequence set is SI, then it is TI for any MPR capability $\gamma$.
\label{Con2}
\end{theorem}
\begin{proof}
From~\eqref{R}, we obtain the result that the throughput $R_i(\tau_1,\tau_2,\ldots,\tau_K)$ can be computed only in terms of $N(b_1,\ldots,b_K|\seq{s}_1^{(\tau_1)},\ldots,\seq{s}_K^{(\tau_K)})$ for some particular choices of $b_1,\ldots,b_K$.
By Theorem~\ref{Con1}, we find each term of the above is a constant function of $\tau_1,\ldots,\tau_K$ if the sequence set is SI.
Thus $R_i(\tau_1,\tau_2,\ldots,\tau_K)$ is also a constant function of $\tau_1,\ldots,\tau_K$, which implies that a SI sequence set must be TI for any $\gamma$. \end{proof}

Theorem~\ref{Con2} implies that we can use known constructions of SI sequences to design TI sequences for any MPR capability $\gamma$.
A general construction of SI sequences was given in~\cite{SCSW}, and we present it here for the sake of completeness.
The proof that the sequences so generated are SI can be found in~\cite{SCSW}.


\smallskip
\noindent {\bf An optimal construction~\cite{SCSW}.}
Given the duty factors $n_1/d_1,n_2/d_2,\ldots,n_K/d_K$ where $\gcd(n_i,d_i) = 1$ for all $i$, we construct $\mathbb{G}_i$, a $\prod_{j=1}^{i-1} d_j \times d_i$ array of zeros and ones such that there are exactly $n_i$ ones in each row. ($\prod_{j=1}^{0} d_j$ is defined as 1, as the empty product is equal to 1 by convention.)
Then construct $\seq{s}_{i}$ by reading out the columns of this array from left to right and extending them periodically to the period $d_1 d_2 \cdots d_K$, for $i=1,2,\ldots,K$.

It is shown that this construction produces the common period $d_1 d_2 \cdots d_K$ for all $K$ sequences, and thus it is {\em an optimal construction of TI sequences} for any $\gamma$ when $\gcd(n_i,d_i) = 1$ for all $i$, in the sense that the period achieves the lower bound in Corollary~\ref{cor:period}.


{\bf Example:} Given the duty factors $f_1=2/3$ and $f_2=f_3=1/3$, we can obtain the following three zero-one arrays by the above construction.
\[
\mathbb{G}_1={\begin{bmatrix}
   1 & 1 & 0
\end{bmatrix}}, \ \
\mathbb{G}_2={\begin{bmatrix}
   1 & 0 & 0  \\
   1 & 0 & 0  \\
   1 & 0 & 0
\end{bmatrix}},  \ \
\mathbb{G}_3 ={\begin{bmatrix}
   1 & 0 & 0  \\
   1 & 0 & 0  \\
   1 & 0 & 0  \\
   1 & 0 & 0  \\
   1 & 0 & 0  \\
   1 & 0 & 0  \\
   1 & 0 & 0  \\
   1 & 0 & 0  \\
   1 & 0 & 0
\end{bmatrix}}.
\]
Then we read out the columns of $\mathbb{G}_i$ from left to right and extend them periodically to $\seq{s}_i$ of length $27$, for $i=1,2,3$.
\begin{align*}
\seq{s}_1 &= [1\ 1\ 0\ 1\ 1\ 0\ 1\ 1\ 0\ 1\ 1\ 0\ 1\ 1\ 0\ 1\ 1\ 0\  1\ 1\ 0\ 1\ 1\ 0\ 1\ 1\ 0] \\
\seq{s}_2 &= [1\ 1\ 1\ 0\ 0\ 0\ 0\ 0\ 0\ 1\ 1\ 1\ 0\ 0\ 0\ 0\ 0\ 0 \ 1\ 1\ 1\ 0\ 0\ 0\ 0\ 0\ 0] \\
\seq{s}_3 &= [1\ 1\ 1\ 1\ 1\ 1\ 1\ 1\ 1\ 0\ 0\ 0\ 0\ 0\ 0\ 0\ 0\ 0 \ 0\ 0\ 0\ 0\ 0\ 0\ 0\ 0\ 0]
\end{align*}
One can check that, for all $\tau_1$, $\tau_2$ and $\tau_3$, the values of the generalized Hamming cross-correlations are $H(\tau_1, \tau_2;(1,2)) = 6$, $H(\tau_2, \tau_3;(2,3)) = 3$, $H(\tau_1, \tau_3;(1,3)) = 6$ and $H(\tau_1, \tau_2, \tau_3;(1,2,3)) = 2$.
Hence the sequence set is SI.
It can produce invariant individual throughput for any MPR capability and thus is also TI.
We have $R_1=8/27, R_2=R_3=1/27$ for $\gamma=1$ and $R_1=16/27, R_2=R_3=7/27$ for $\gamma=2$ which are both in accordance with Theorem~\ref{T}.
Its period is $27$ which achieves the lower bound presented in Corollary~\ref{cor:period} for $\gamma=1,2$.

We furthermore present a numerical evaluation of the throughput performance of the TI sequences produced by the SI construction.
A symmetric and saturated system with the MPR capability $\gamma$ is considered.
We conduct $10^5$ simulation runs for each $K,\gamma$ to generate $10^5$ combinations of uniformly distributed relative shifts.
Similar to~\cite{SCSW} for $\gamma=1$, in order to examine the throughput variation and average throughput over the sequence period, TI sequences with the duty factor $1/K$ are compared with a random access scheme in which a user sends a packet in each slot with an independent probability $1/K$.
In Fig.~\ref{fig:num2}, we plot the maximum, mean and minimum individual throughput for $\gamma=2,3$ with $\gamma < K\leq 7$.
TI sequences yield constant symmetric individual throughput as derived in~\eqref{eq:T}.
For the random access scheme, the mean throughput is equal to that of TI sequences as expected; the maximum and minimum throughput are getting closer when the averaging time scale increases, due to the strong law of large numbers.
Results have shown that the SI construction can provide zero-variance on throughput for MPR.


\begin{figure}[htbp]
\centering
\includegraphics[height=3.4in,width=3.2in]{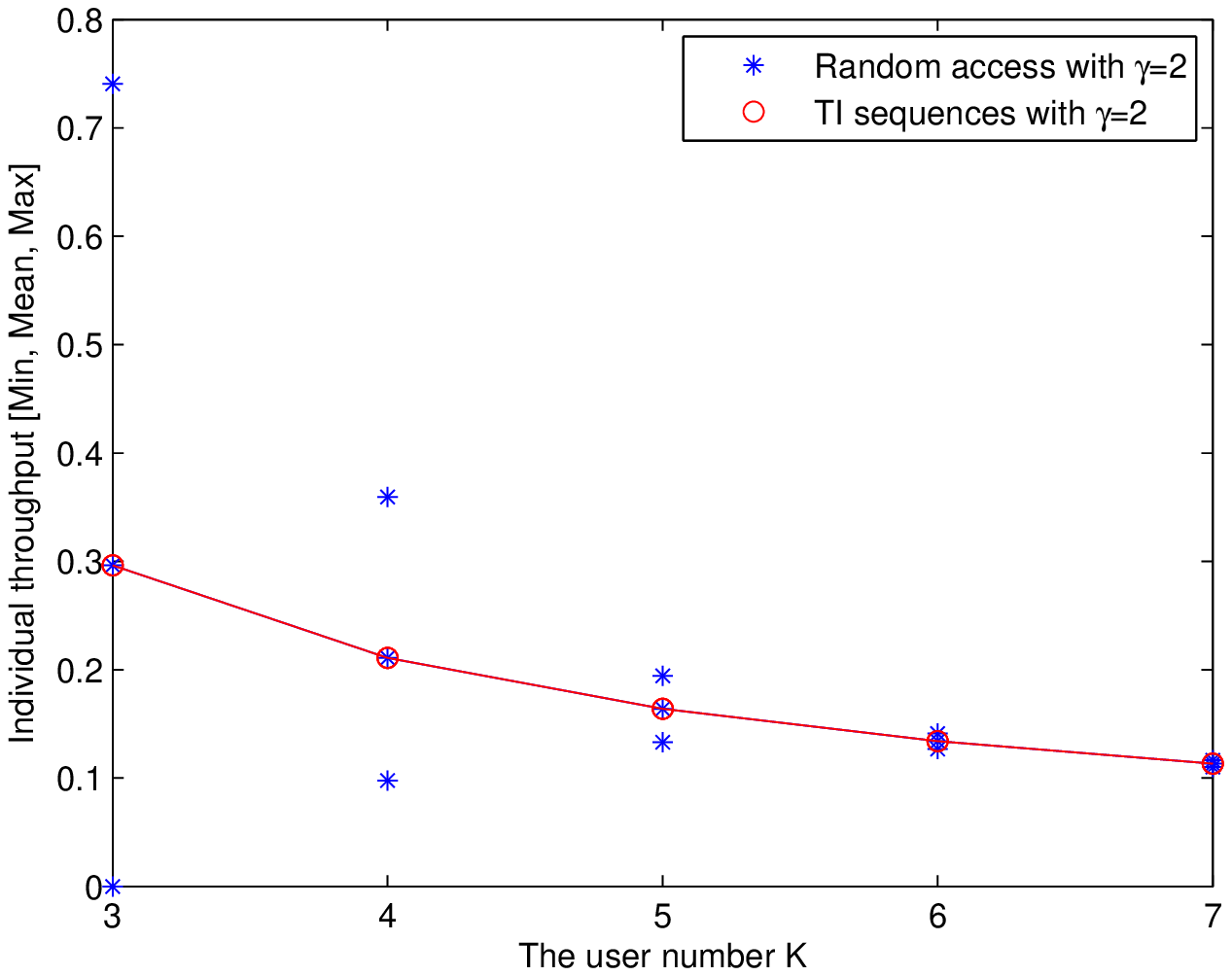}
\includegraphics[height=3.4in,width=3.2in]{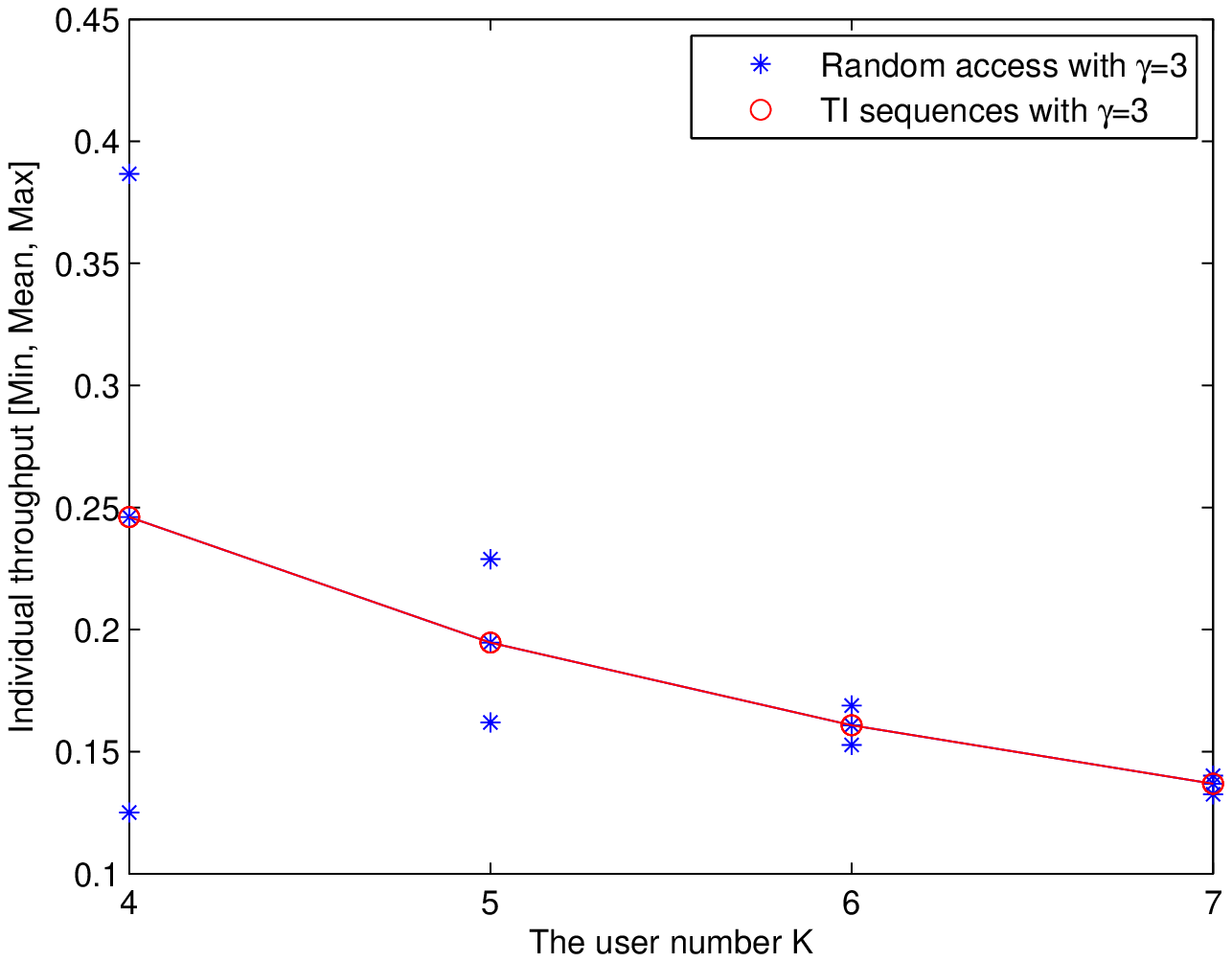}
\caption{The minimum, mean, and maximum individual throughput from simulation for $\gamma=2,3$ with $\gamma < K\leq 7$.
The mean value is connected by a piece-wise linear curve.
The symbols above and below this curve indicate the maximum and minimum value, respectively.}
\label{fig:num2}
\end{figure}

\section{Identification and Decoding}
 In a protocol sequences based multiaccess protocol, the receiver needs to accomplish the following two tasks~\cite{Nguyen92,GyorfiVajda93,Wong07}:
\begin{enumerate}
  \item identify the sender of each successfully received packet (the {\em identification problem});
  \item decode each successfully received packet and recover the original information (the {\em decoding problem}).
\end{enumerate}

Merely relying on some special structure of SI sequences, Shum et al. in~\cite{SCSW} generalized the {\em decimation algorithm}~\cite{Massey85,Rocha00} to solve the identification problem of all uncollided packets without the need of using header information.
However, this solution is invalid for packets survived in collisions if MPR is considered for $\gamma>1$.

The problem of packet separation at the PHY layer with MPR is usually formulated as signal separation in a {\em multiple-input-multiple-output} (MIMO) system.
Let $x_i(n)$ denote the symbol transmitted by user $i$ in symbol duration $n$,
$\mathbf{w}(n)$ denote the additive noise,
and $\mathbf{h}_{i}(n)$ denote the user $i$'s channel vector in symbol duration $n$.
The vector $\mathbf{h}_{i}(n)$ depends on the configuration of channel diversity, for example, in a multiple antenna system its $m$th element represents the channel coefficient from user $i$ to the $m$-th receive antenna.
Suppose that the users $i_1,i_2,\ldots, i_M$ are transmitting simultaneously in symbol duration $n$, and then the received signal at the receiver in symbol duration $n$ is given by
\begin{align*}
\mathbf{y}(n)&=\sum_{j=1}^M \mathbf{h}_{i_j}(n)x_{i_j}(n)+\mathbf{w}(n) \\
             &=\mathbf{H}(n)\mathbf{x}(n)+\mathbf{w}(n),
\end{align*}
where $\mathbf{H}(n)=[\mathbf{h}_{i_1}(n), \mathbf{h}_{i_2}(n),\ldots,\mathbf{h}_{i_M}(n)]$, $\mathbf{x}(n)=[{x}_{i_1}(n), {x}_{i_2}(n),\ldots,{x}_{i_M}(n)]^T$.
The problem of signal separation here is how to estimate transmitted symbols $\mathbf{x}(n)$ from the received vector $\mathbf{y}(n)$, when at most $\gamma$ users are transmitting simultaneously, i.e., $M\leq \gamma$.
It is unrealistic to assume that the receiver knows the channel vector $\mathbf{H}(n)$, and thus various training-based signal separation algorithms have been developed in the literature for the estimation of $\mathbf{H}(n)$, which requires that the receiver has priori knowledge of the senders' identities, the training sequences and their locations in a packet.
However, these algorithms cannot be applied in protocol sequences based multiaccess, because that the transmitting users are changing and unknown to the receiver due to the random relative shifts of the users.

To address these problems, this section proposes a mechanism of the identification and decoding for TI sequences considering $\gamma>1$, as below.
\begin{enumerate}[(a)]
  \item Each packet contains a header that indicates the user identity and packet identity.
   This is also a common practice in protocol sequences systems~\cite{Wong07,CRT}.
   Note that the packet identity only has one bit information which denotes whether this packet is transmitted in an odd period (the period order number module 2 equals 1).
   The bit size of such a header can be found as $1+\log_2 K$, so its effect on the system performance is negligible.
  \item For those collision slots occupied by at most $\gamma$ transmitting users, one can find out the packet content by employing blind signal separation algorithms~\cite{Blind96,Blind98,Tong01}, which can obtain $\mathbf{x}(n)$ from $\mathbf{y}(n)$ without knowing who are the senders and $\mathbf{H}(n)$ in advance.
Interested readers are referred to~\cite{Blind96,Blind98} for more details at the signal processing level.
The complexity of the separation algorithm here depends on $\gamma$, since $\gamma$ is the maximum dimension of $\mathbf{x}(n)$.
Note that such an idea was also adopted in~\cite{ZZL09} for packet separation in IEEE 802.11 with MPR.
\item
As mentioned in~\cite{Massey85,Wong07,CRT}, we apply a Reed-Solomon (RS) code across the data payloads in a period to jointly code and decode the original information even if some packets are unavoidably lost, provided that the required minimum number of survived packets in each period of a user is equal to or smaller than that guaranteed by TI sequences derived in Section III.
The receiver knows which packets are transmitted in the same sequence period through the packet identity.
\end{enumerate}
An illustration of the packet format and RS code is given in Fig.~\ref{fig:decode}.
By using TI sequences, we ensure that in each sequence period there are enough time slots in which the packets can be separated.
Moreover, by using the embedded user identity and packet identity, the receiver can further recover enough fragments of the coded data to decode the original information.
\begin{figure}
\begin{center}
  \includegraphics[width=4.8in]{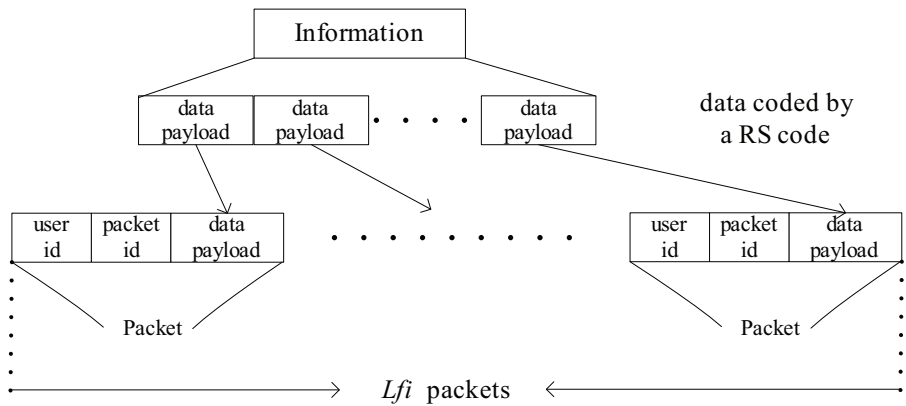}
\end{center}
\caption{Packet format using RS code of user $i$ with the duty factor $f_i$ for $i=1,2,\ldots,K$.}
\label{fig:decode}
\end{figure}

One sees that all approaches mentioned above have been commonly employed in protocol sequences for $\gamma=1$ and contention based network for $\gamma>1$.
Therefore, our mechanism does not cause any additional burden on system performance and receiver design, compared with other MPR protocols.





\section{Structural Theorem}
It was proved in~\cite{SCSW} that the period of $K$ SI sequences with the duty factors $f_i=n_i/d_i$ for $i=1,2,\ldots,K$ must be divisible by
\begin{equation}
\frac{\prod_{i \in \mathcal{U}}d_i}{\gcd(\prod_{i \in \mathcal{U}}d_i,\prod_{i \in \mathcal{U}}n_i)},
\label{eq:exp}
\end{equation}
for any subset $\mathcal{U}$ of $\mathcal{K}$.

Given that SI sequences possess the TI property for any MPR capability and suffer from the drawback that the common period grows exponentially with $K$ as shown in~\eqref{eq:exp}, a natural question is whether there are shorter TI sequences, as a smaller period is favorable in practical applications.
Unfortunately, results we have attained in Section IV rule this out for any $\gamma$ under the assumption that $\gcd(n_i,d_i) = 1$ in $f_i=n_i/d_i$ for all $i$.
However, there is no definite proof yet that the lower bound on the minimum period of TI sequences grows exponentially.
So technically it is of interest to know more about the structure of TI sequences and resolve whether there exist TI sequences which are not SI.
Results in Section VII can shed light in this direction.

The main objective of this section is to demonstrate that the SI property is intrinsic for TI sequences by proving TI implies SI for many specific cases.
This tends to imply that there are no practical solutions to the TI problem with MPR.
We may need to use wobbling sequences~\cite{Wong07} or CRT sequences~\cite{CRT} instead.
Although these sequences are not SI, their pairwise Hamming cross-correlation functions are close to being constant and they can guarantee a quasi-TI throughput performance on a relatively short time scale.

Define a {\em $k$-subset} as a set containing exactly $k$ elements.
We begin our study on the structural theorem with the following proposition.

\begin{prop}
Let $\gamma$ be a positive integer smaller than $K$ and $\mathcal{S}=\{\seq{s}_1,\ldots,\seq{s}_K\}$ be a set of $K$ binary sequences.
If $\mathcal{S}$ is TI with the MPR capability $\gamma$ and any $(K-1)$-subset of $\mathcal{S}$ is SI, then $\mathcal{S}$ is SI.
\label{prop:K-1toK}
\end{prop}
\begin{proof}
For a given set of relative shifts $\tau_1,\ldots,\tau_K$, define a $K\times L$ matrix $\mathbb{M}$ by putting $\seq{s}_i^{(\tau_i)}$ at the $i$-th row, for $i=1,\ldots,K$.
By regarding an ordered tuple $\mathsf{A}\in\mathcal{O}_K$ as a subset of $\mathcal{K}$, the generalized Hamming cross-correlation $H(\tau_j:\,j\in\mathsf{A};\mathsf{A})$ can be viewed as the number of all-one columns in the submatrix obtained by collecting all $j$-th row of $\mathbb{M}$, $j\in\mathsf{A}$.
Notice that $L R_i(\tau_1,\ldots,\tau_K)$ counts those `1's on the $i$-th row such that the corresponding column-sum in $\mathbb{M}$ does not exceed $\gamma$.
By the principle of inclusion-and-exclusion, we have
\begin{align*}
L R_i(\tau_1,\ldots,\tau_K) &=
Lf_i - \sum_{i\in\mathsf{A},|\mathsf{A}|=\gamma+1} H(\tau_j:j\in \mathsf{A}; \mathsf{A}) \\
&+ \sum_{i\in\mathsf{A},|\mathsf{A}|=\gamma+2} H(\tau_j:j\in \mathsf{A}; \mathsf{A}) + \cdots + (-1)^{K-\gamma} \sum_{i\in\mathsf{A},|\mathsf{A}|=K } H(\tau_j:j\in \mathsf{A}; \mathsf{A}).
\end{align*}
Since $H(\tau_j:j\in \mathsf{A}; \mathsf{A})$ is SI for $|\mathsf{A}|<K$ and $R_i(\tau_1,\ldots,\tau_K)$ has zero-variance by the condition, $H(\tau_j:j\in \mathcal{K}; \mathcal{K})$ is also SI.
Therefore, $H(\tau_j:j\in \mathsf{A}; \mathsf{A})$ is SI for every $\mathsf{A}$ in $\mathcal{K}$, which implies the entire sequence set $\mathcal{S}$ is SI.
\end{proof}

\begin{definition} \label{def:theta_change}
Consider $M$ binary sequences $\seq{s}_1,\ldots,\seq{s}_M$ with a common period.
Let $\mathsf{M}=(1,\ldots,M)$.
Given two distinct systems of relative shifts $\tau=(\tau_1,\ldots,\tau_M)$ and $\tau'=(\tau'_1,\ldots,\tau'_M)$, let
\[
\delta_i(\tau\to\tau';\mathsf{M}) := \theta_i(\tau_1', \ldots, \tau_M';\mathsf{M}) - \theta_i(\tau_1, \ldots, \tau_M;\mathsf{M}),
\]
denote the change of the value $\theta_i$, for $i=0,1,\ldots,M$.
\end{definition}

\begin{lemma}
Consider $M$ binary sequences $\seq{s}_1,\ldots,\seq{s}_M$ with a common period and two distinct systems of relative shifts $\tau=(\tau_1,\ldots,\tau_M)$ and $\tau'=(\tau'_1,\ldots,\tau'_M)$.
If any $(M-1)$-subset of $\{\seq{s}_1,\ldots,\seq{s}_M\}$ is SI, then for $i=1,\ldots,M-1$,
\begin{equation}
\delta_i(\tau\to\tau';\mathsf{M}) = (-1)^{M-i}{M\choose{i}}\delta_M(\tau\to\tau';\mathsf{M}).
\end{equation}
\label{lem:delta}
\end{lemma}
\begin{proof}
In this proof we simply write $\delta_i$ instead of $\delta_i(\tau\to\tau';\mathsf{M})$.
For $j=1,\ldots,M-1$, let $\Phi_j$ be a collection of all ordered pairs $(\mathsf{A};t)$ with $|\mathsf{A}|=j$, $\mathsf{A}\in\mathcal{O}_M$, such that ${s}_{i}(t+\tau_{i})=1$ for each $i\in\mathsf{A}$.
We shall count the cardinality of $\Phi_j$ in two ways.
By the definition of $\theta_i$, we have
\[
|\Phi_j| = \sum_{i=j}^M {i\choose j}\theta_i(\tau_1, \ldots, \tau_M;\mathsf{M}).
\]
On the other hand, $\Phi_j$ computes all generalized Hamming cross-correlations among every $j$ sequences.
Then, we have
\[
|\Phi_j| = \sum_{|\mathsf{A}|=j,\mathsf{A}\in\mathcal{O}_M} H(\tau_i:\,i\in\mathsf{A};\mathsf{A}),
\]
which is a constant function of $\tau_1,\ldots,\tau_j$ due to the assumption that each $(M-1)$-subset is SI.
Therefore, we have
\[
\sum_{i=j}^M {i\choose j}\theta_i(\tau_1,\ldots,\tau_M;\mathsf{M}) = \sum_{i=j}^M {i\choose j}\theta_i(\tau_1',\ldots,\tau_M';\mathsf{M})
\]
and thus
\begin{equation}
\sum_{i=j}^M {i\choose j}\delta_i=0.
\label{eq:delta}
\end{equation}
Finally, the result follows by plugging $j=M-1,M-2,\ldots,1$ into (\ref{eq:delta}) step by step.
More precisely, from \eqref{eq:delta} we inductively have
\begin{align*}
\delta_j &= -{j+1\choose j}\delta_{j+1} -{j+2\choose j}\delta_{j+2} - \cdots -{M\choose j}\delta_{M} \\
&= (-1)^{M-j}\left[ {j+1\choose 1}{M\choose j+1} - {j+2\choose 2}{M\choose j+2} + \cdots + (-1)^{M-j-1}{M\choose M-j}{M\choose M}\right]\delta_M \\
&= (-1)^{M-j}{M\choose j}\left[ {M-j\choose 1} - {M-j\choose 2} + \cdots + (-1)^{M-j-1}{M-j\choose M-j} \right]\delta_M \\
&= (-1)^{M-j} {M\choose j}\delta_M,
\end{align*}
as desired.
\end{proof}

\begin{lemma} \label{lem:theta_less_than}
Let $\mathcal{S}=\{\seq{s}_1,\ldots,\seq{s}_K\}$ be a set of $K$ binary sequences with a common period $L$.
Let $\mathsf{A}=(1,\ldots,M)$, $\mathsf{B}=(M+1,\ldots,K)$, for some $M<K$, and $\tau=(\tau_1,\ldots,\tau_M)$, $\tau'=(\tau'_1,\ldots,\tau'_M)$, $\tau^*=(\tau^*_1,\ldots,\tau^*_M)$ be three systems of relative shifts.
If $\mathcal{S}$ is TI with the MPR capability $\gamma$ and any $(M-1)$-subset of $\mathcal{S}$ is SI, then
\begin{equation} \label{eq:delta_theta}
\delta_M(\tau\to\tau';\mathsf{A}) \sum_{i=1}^{M-1}(-1)^{M-i}{M-2\choose {i-1}} \theta_{\gamma-i}(\tau_{M+1}^*, \ldots, \tau_K^*;\mathsf{B})=0.
\end{equation}
Note that $\theta_{\gamma-i}(\tau_{M+1}^*,\ldots,\tau_K^*;\mathsf{B})=0$ if $\gamma < i$ or $\gamma-i>K-M$.
\end{lemma}
\begin{proof}
Let $R_i$ be the individual throughput associated with $\seq{s}_i$.
Similar to the arguments in (\ref{eq:average_1}) -- (\ref{eq:average_2}), after taking the expectation over the relative shifts $\tau_1,\ldots,\tau_M$, we have
\begin{equation}
\mathbb{E}(LR_1+\cdots+LR_M)  =\sum_{i=1}^M \frac{i\,\theta_i(\tau_1,\ldots,\tau_M;\mathsf{A})\theta_{\leq\gamma-i}(\tau_{M+1}^*,\ldots,\tau_K^*;\mathsf{B})}{L}, \label{eq:average_M}
\end{equation}
which is a constant function of $\tau_1,\ldots,\tau_M$ due to the zero-variance of $R_1+\cdots+R_M$.

Now, consider that the relative shifts of the sequences $\seq{s}_1,\ldots,\seq{s}_M$ are changed from $\tau$ to $\tau'$.
By Lemma~\ref{lem:delta} and (\ref{eq:average_M}), we have
\begin{align*}
0 &= \sum_{i=1}^M \frac{i\,\delta_i(\tau\to\tau';\mathsf{A})\,\theta_{\leq\gamma-i}(\tau_{M+1}^*,\ldots,\tau_K^*;\mathsf{B})}{L} \\
  &= \frac{1}{L}\sum_{i=1}^M (-1)^{M-i}\, i\,{M\choose{i}}\delta_M(\tau\to\tau';\mathsf{A}) \theta_{\leq\gamma-i}(\tau_{M+1}^*,\ldots,\tau_K^*;\mathsf{B}) \\
  &= \frac{\delta_M(\tau\to\tau';\mathsf{A}) M}{L} \sum_{i=1}^M (-1)^{M-i} {M-1\choose{i-1}} \theta_{\leq\gamma-i}(\tau_{M+1}^*,\ldots,\tau_K^*;\mathsf{B}).
\end{align*}
Thus,
\begin{equation} \label{eq:delta_theta_1}
\delta_M(\tau\to\tau';\mathsf{A}) \sum_{i=1}^{M}(-1)^{M-i}{M-1\choose {i-1}} \theta_{\leq\gamma-i}(\tau_{M+1}^*, \ldots, \tau_K^*;\mathsf{B})=0.
\end{equation}
Finally, the target identity \eqref{eq:delta_theta} can be obtained from \eqref{eq:delta_theta_1}.
See Appendix~\ref{sec:app_2} for the derivation steps of \eqref{eq:delta_theta_1} $\Rightarrow$ \eqref{eq:delta_theta}.
\end{proof}

The following result follows from the previous lemma.
\begin{prop}
Following the setting of Lemma~\ref{lem:theta_less_than}, if all sequences have the same duty factor $f$, then
\begin{align*}
\delta_M(\tau\to\tau';\mathsf{A}) \sum_{i=1}^{M-1}(-1)^{M-i}{M-2\choose {i-1}} {K-M\choose {\gamma-i}} f^{\gamma-i}(1-f)&^{K-M-\gamma+i}=0.
\end{align*}
Note that ${K-M\choose {\gamma-i}}=0$ if $\gamma < i$ or $\gamma-i>K-M$.
\label{cor:constant_f}
\end{prop}
\begin{proof}
Observe that there is no constraint on the choice of $\tau_{M+1}^*, \ldots, \tau_K^*$ in \eqref{eq:delta_theta}.
After taking the expectation over all possible relative shifts of $\seq{s}_{M+1},\ldots,\seq{s}_K$, we can replace the term $\theta_{\gamma-i}(\tau_{M+1}^*, \ldots, \tau_K^*;\mathsf{B})$ in \eqref{eq:delta_theta} by
$${K-M\choose {\gamma-i}} f^{\gamma-i}(1-f)^{K-M-\gamma+i}.$$
This completes the proof.
\end{proof}


\smallskip
We are ready for our main results in this section.
\subsection{The case of $\gamma=1$}

\begin{theorem} \label{thm:mustbe SI 1}
Let $\mathcal{S}=\{\seq{s}_1,\ldots,\seq{s}_K\}$ be a set of $K$ binary sequences.
If $\mathcal{S}$ is TI with the MPR capability $1$, then it is SI.
\end{theorem}
\begin{proof}
We claim by induction that any $M$-subset of $\mathcal{S}$ is SI for $M=2,3,\ldots,K-1$.
The case of $M=2$ has been settled in Theorem~\ref{pSI}, so we consider $M\geq 3$ and assume that any set of $M-1$ sequences is SI.

Without loss of generality, let $\mathsf{A}=(1,\ldots,M)$ and $\mathsf{B}=(M+1,\ldots,K)$.
Let the relative shifts of $\seq{s}_1,\ldots,\seq{s}_M$ be changed from $\tau=(\tau_1,\ldots,\tau_M)$ to $\tau'=(\tau_1',\ldots,\tau_M')$.
Now, fix the relative shifts of $\seq{s}_{M+1},\ldots,\seq{s}_K$ be $\tau^*=(\tau^*_{M+1},\ldots,\tau^*_K)$ such that $\theta_0(\tau^*_{M+1},\ldots,\tau^*_K;\mathsf{B})>0$.
Such $\tau^*$ is always existent because there is no all-one sequence in $\mathcal{S}$ for $\gamma=1$ following Proposition~\ref{prop:allone}.
By Lemma~\ref{lem:theta_less_than} ($\gamma=1$), we have
$$\delta_M(\tau\to\tau';\mathsf{A})\theta_0(\tau^*_{M+1},\ldots,\tau^*_K;\mathsf{B})=0,$$
which implies that $\delta_M(\tau\to\tau';\mathsf{A})=0$ due to the choice of $\tau^*$.
Therefore, $\theta_M(\tau_1,\ldots,\tau_M;\mathsf{A})$ is a constant function of $\tau_1,\ldots,\tau_M$, and thus $\{\seq{s}_1,\ldots,\seq{s}_M\}$ is SI.

By induction on $M$ from $M=2$, we conclude that any $(K-1)$-subset of $\mathcal{S}$ is SI.
This furthermore implies that the entire set $\mathcal{S}$ is SI from Proposition~\ref{prop:K-1toK}.
\end{proof}

Theorem~\ref{thm:mustbe SI 1} asserts that SI and TI sequences are equivalent for $\gamma=1$ (i.e., SPR).

\subsection{The case of $\gamma=2$}
\begin{theorem}
Let $\mathcal{S}=\{\seq{s}_1,\ldots,\seq{s}_K\}$ be a set of $K$ binary sequences with the same duty factor $f=\frac{n}{d}\neq 0,1$, which is TI with the MPR capability 2.
$\mathcal{S}$ is SI if $\gcd(K-2,d)=1$.
\label{thm:mustbe SI 2}
\end{theorem}
\begin{proof}
We claim by induction that any $M$-subset of $\mathcal{S}$ is SI for $M=2,3,\ldots,K-1$.
The case of $M=2$ holds by Theorem~\ref{pSI}, so we consider $M\geq 3$.

Similar to the proof of Theorem~\ref{thm:mustbe SI 1}, let $\mathsf{A}=(1,\ldots,M)$, $\mathsf{B}=(M+1,\ldots,K)$, $\tau=(\tau_1,\ldots,\tau_M)$ and $\tau'=(\tau_1',\ldots,\tau_M')$.
By Proposition~\ref{cor:constant_f} ($\gamma=2$), we have
\[
\delta_M(\tau\to\tau';\mathsf{A})\left( (K-M)f(1-f)^{K-M-1} - (M-2) (1-f)^{K-M} \right) = 0.
\]
By assuming $K\geq 3$, the above equation can be simplified to
$$\delta_M(\tau\to\tau';\mathsf{A})\left( f - \frac{M-2}{K-2} \right) = 0.$$
There are two cases.
\begin{enumerate}[(a)]
\item $f=\frac{M-2}{K-2}$: If $\gcd(K-2,d)=1$, then $f=\frac{M-2}{K-2}$ contradicts to the assumption that $f=\frac{n}{d}$, and thus we have $\delta_M(\tau\to\tau';\mathsf{A})=0$.
\item $\delta_M(\tau\to\tau';\mathsf{A})=0$: This implies that $\theta_M(\tau_1,\ldots,\tau_M;\mathsf{A})$ is a constant of $\tau_1,\ldots,\tau_M$.
    Hence, $\{\seq{s}_1,\ldots,\seq{s}_M\}$ is SI.
\end{enumerate}

By induction on $M$ from $M=2$, we conclude that any $(K-1)$-subset of $\mathcal{S}$ is SI.
This furthermore implies that the entire set $\mathcal{S}$ is SI from Proposition~\ref{prop:K-1toK}.
\end{proof}

\subsection{The case of $\gamma=3$}

\begin{theorem}
Let $\mathcal{S}=\{\seq{s}_1,\ldots,\seq{s}_K\}$ be a set of $K$ binary sequences with the same duty factor $f=\frac{n}{d}\neq 0,1$, which is TI with the MPR capability 3.
Then $\mathcal{S}$ is SI if (i) $K-3$ is prime, and (ii) $\gcd(\frac{K-2}{2},d)=1$.
\label{thm:mustbe SI 3}
\end{theorem}
\begin{proof}
The structure of this proof is similar to Theorems~\ref{thm:mustbe SI 1} and \ref{thm:mustbe SI 2}, so we only show the inductive step here.

Similar to the proof of Theorem~\ref{thm:mustbe SI 1}, let $\mathsf{A}=(1,\ldots,M)$, $\mathsf{B}=(M+1,\ldots,K)$, $\tau=(\tau_1,\ldots,\tau_M)$ and $\tau'=(\tau_1',\ldots,\tau_M')$.
By Proposition~\ref{cor:constant_f} ($\gamma=3$), we have
\begin{equation}
\delta_M(\tau\to\tau';\mathsf{A})\left({K-2\choose 2}f^2 - (K-3)(M-2)f + {M-2\choose 2}\right) = 0.
\label{eq:constant_f_gamma=3}
\end{equation}
Suppose to the contrary that $\delta_M(\tau\to\tau';\mathsf{A})\neq 0$.
If $M=3$, since $K>M$ and $f\neq 0$, then \eqref{eq:constant_f_gamma=3} holds only when $f=\frac{2}{K-2}$, which contradicts to the assumption that 
$f=\frac{n}{d}$ and (ii) $\gcd(\frac{K-2}{2},d)=1$ ($\frac{K-2}{2}$ must be an integer as $K-3$ is prime).
If $M>3$, \eqref{eq:constant_f_gamma=3} can be simplified to
\begin{equation} \label{eq:quadratic_equation}
{p+1\choose 2}f^2 - p(M-2)f + {M-2\choose 2} = 0,
\end{equation}
where $p=K-3$ is a prime number by (i).
Since the duty factor $f$ should be a rational number, the discriminant
\[
D:=p^2(M-2)^2-p(p-1)(M-2)(M-3)
\]
of the quadratic equation \eqref{eq:quadratic_equation} is a square number.
That is,
\begin{equation} \label{eq:discriminant_condition}
(M-2)(M-3)\equiv 0 \text{ (mod $p$).}
\end{equation}
Since $M<K$, \eqref{eq:discriminant_condition} holds only when $M-2=p$.
This implies that $f=1$ or $\frac{p-1}{p+1}$.
The former solution contradicts to the original assumption, while the latter one contradicts to $\gcd(\frac{K-2}{2},d)=1$.
All of above promise that $\delta_M(\tau\to\tau';\mathsf{A})=0$.
Therefore, $\theta_M(\tau_1,\ldots,\tau_M;\mathsf{A})$ is a constant of $\tau_1,\ldots,\tau_M$.
Hence, $\{\seq{s}_1,\ldots,\seq{s}_M\}$ is SI. This completes the inductive step.
\end{proof}

\subsection{The case of $\gamma=K-1,K-2,K-3$}
For larger $\gamma$, we first rewrite \eqref{eq:average_M} as
\begin{align*}
\mathbb{E}&(LR_1+\cdots+LR_M) \notag = \sum_{i=1}^{M} Lf_i - \sum_{i=1}^M \frac{i\,\theta_i(\tau_1,\ldots,\tau_M;\mathsf{A})\theta_{\geq\gamma+1-i}(\tau_{M+1}^*,\ldots,\tau_K^*;\mathsf{B})}{L},
\end{align*}
where $\theta_{\geq j}(\tau_{M+1}^*,\ldots,\tau_K^*;\mathsf{B}):=\sum_{k=j}^{K-M}\theta_{k}(\tau_{M+1}^*,\ldots,\tau_K^*;\mathsf{B})$.
By the same argument in previous subsections, we obtain the following results in parallel with Theorems~\ref{thm:mustbe SI 1}, \ref{thm:mustbe SI 2} and \ref{thm:mustbe SI 3}.
The proof is identical as before and is omitted here.

\begin{theorem}
If a set of $K$ binary sequences is TI with MPR capability $K-1$, then this sequence set is SI.
\label{thm:mustbe SI -1}
\end{theorem}

\begin{theorem}
Let $\mathcal{S}$ be a set of $K$ binary sequences with the same duty factor $f=\frac{n}{d}\neq 0,1$, which is TI with MPR capability $\gamma$.
\begin{enumerate}[(a)]
\item ($\gamma=K-2$.) $\mathcal{S}$ is SI if $\gcd(K-2,d)=1$.
\item ($\gamma=K-3$.) $\mathcal{S}$ is SI if (i) $K-3$ is prime, and (ii) $\gcd(\frac{K-2}{2},d)=1$.
\end{enumerate}
\end{theorem}

\section{Conclusion}
In this paper, considering the MPR capability $\gamma$ we investigate TI sequences, which produce the invariant throughput for any relative shifts.
Only one previous known result on TI sequences is that SI sequences must be TI for $\gamma=1$~\cite{SCSW}.
Considering some specific form of duty factors, this paper obtains that the length of TI sequences must be exponential in the number of users for any $\gamma$, and proves that some known constructions of SI sequences can be used to design optimal TI sequences for any $\gamma$.
In addition, we explore the bit structure of TI sequences by showing that they must be pairwise SI for any $\gamma$, and further be SI in many specific cases.
This tends to indicate that there are no shorter solutions to the TI problems than SI sequences.
To our knowledge, the existence of TI sequences which are not SI is still unknown.
Another aspect is to apply some known shorter sequences, such as wobbling sequences or CRT sequences, to promise a quasi-TI performance for MPR, which is of more practical interests in a realistic system.
We leave these problems to the interested readers.
Furthermore, having understood the fundamental behavior of MPR on TI sequences, we propose a practical identification and decoding mechanism, by incorporating packet header, RS code, and advanced PHY-layer blind packet separation algorithms.

\appendix

\subsection{Proof of (\ref{eq:average_2})} \label{sec:app_1}

We first construct the following four binary sequences:
\begin{enumerate}
  \item $\seq{s}_{\alpha_1}(t)=1$ if and only if $s_1(t+\tau_1)+s_2(t+\tau_2)=1$ for each $0\leq t \leq L-1$;
  \item $\seq{s}_{\alpha_2}(t)=1$ if and only if $s_1(t+\tau_1)=s_2(t+\tau_2)=1$ for each $0\leq t \leq L-1$;
  \item $\seq{s}_{\beta_1}(t)=1$ if and only if $\sum_{i=3}^{K} s_i(t+\tau^*_i) \leq \gamma-1$ for each $0\leq t \leq L-1$.
  \item $\seq{s}_{\beta_2}(t)=1$ if and only if $\sum_{i=3}^{K} s_i(t+\tau^*_i) \leq \gamma-2$ for each $0\leq t \leq L-1$.
\end{enumerate}
It is obvious that
\begin{align*}
& N(1|\seq{s}_{\alpha_1})=\theta_1(\tau_1,\tau_2;\mathsf{A}),\ \ \ \ \ \ \ \ \ \ \ \ \ \ N(1|\seq{s}_{\alpha_2})=\theta_2(\tau_1,\tau_2;\mathsf{A}),\\
& N(1|\seq{s}_{\beta_1})=\theta_{\leq \gamma-1}(\tau^*_3,\ldots,\tau^*_K;\mathsf{B}), \ \ \ \ N(1|\seq{s}_{\beta_2})=\theta_{\leq \gamma-2}(\tau^*_3,\ldots,\tau^*_K;\mathsf{B}).
\end{align*}
Then by \eqref{eq:Shum09} we have
\[
\sum_{\tau_{\alpha_1}=0}^{L-1} \sum_{\tau_{\beta_1}=0}^{L-1} N(1,1|\seq{s}_{\alpha_1}^{(\tau_{\alpha_1})},\seq{s}_{\beta_1}^{(\tau_{\beta_1})})=LN(1|\seq{s}_{\alpha_1})N(1|\seq{s}_{\beta_1}).
\]
The left-hand-side in the above identity is equal to $L^2\mathbb{E}(T_1(\tau_1,\tau_2))$.
Then we have $\mathbb{E}(T_1(\tau_1,\tau_2)) = \theta_1(\tau_1,\tau_2;\mathsf{A}) \theta_{\leq \gamma-1}(\tau^*_3,\ldots,\tau^*_K;\mathsf{B})/L$.

Similarly, we have
\[
\sum_{\tau_{\alpha_2}=0}^{L-1} \sum_{\tau_{\beta_2}=0}^{L-1} N(1,1|\seq{s}_{\alpha_2}^{(\tau_{\alpha_2})},\seq{s}_{\beta_2}^{(\tau_{\beta_2})})=LN(1|\seq{s}_{\alpha_2})N(1|\seq{s}_{\beta_2}),
\]
and thus $\mathbb{E}(T_2(\tau_1,\tau_2))= \theta_2(\tau_1,\tau_2;\mathsf{A}) \theta_{\leq \gamma-2}(\tau^*_3,\ldots,\tau^*_K;\mathsf{B})/L$.

\subsection{Proof of \eqref{eq:delta_theta_1} $\Rightarrow$ \eqref{eq:delta_theta}} \label{sec:app_2}

For convenience, in this proof $\delta_M(\tau\to\tau';\mathsf{A})$, $\theta_{k}(\tau^*_{M+1},\ldots,\tau^*_K;\mathsf{B})$ and $\theta_{\leq k}(\tau^*_{M+1},\ldots,\tau^*_K;\mathsf{B})$ will be abbreviated as $\delta_M$, $\theta_{k}$ and $\theta_{\leq k}$, respectively.
By Lemma~\ref{lem:theta_less_than},
\begin{align*}
0 =& ~\delta_M \sum_{i=1}^{M}(-1)^{M-i}{M-1\choose {i-1}} \theta_{\leq\gamma-i} \\
  =& ~\delta_M (-1)^{M-1} \sum_{i=1}^{M} (-1)^{i-1} {M-1\choose {i-1}} \theta_{\leq\gamma-i} \\
  =& ~\delta_M (-1)^{M-1} \sum_{i=1}^{M-1}(-1)^{i-1} \mathbf{A}_i \,(\theta_{\leq\gamma-i}-\theta_{\leq\gamma-i-1}) + \delta_M \,\mathbf{B} \,\theta_{\leq\gamma-M},
\end{align*}
where
\begin{align*}
\mathbf{A}_i = \sum_{j=1}^{i}(-1)^{j-1}{M-1\choose {i-j}},\, \mathbf{B} = \sum_{j=0}^{M-1}(-1)^j {M-1\choose {M-1-j}}.
\end{align*}
Binomial recursive formula states that
$${n\choose m} - {n-1\choose {m-1}} = {n-1\choose m}, \, \forall 1\leq m\leq n.$$
By replacing ${M-1\choose 0}$ by ${M-2\choose 0}$, we have $\mathbf{A}_i={M-2\choose {i-1}}$.
In addition, $\mathbf{B}=0$ by binomial theorem.
The above equation is therefore
\begin{align*}
\delta_M &\sum_{i=1}^{M-1} (-1)^{M-i} {M-2\choose{i-1}} \left(\theta_{\leq\gamma-i} - \theta_{\leq\gamma-i-1}\right) = \delta_M \sum_{i=1}^{M-1} (-1)^{M-i} {M-2\choose{i-1}} \theta_{\gamma-i}.
\end{align*}

\bibliographystyle{IEEEtran}
\bibliography{UI}

\begin{thebibliography}{10}
\providecommand{\url}[1]{#1}
\csname url@samestyle\endcsname
\providecommand{\newblock}{\relax}
\providecommand{\bibinfo}[2]{#2}
\providecommand{\BIBentrySTDinterwordspacing}{\spaceskip=0pt\relax}
\providecommand{\BIBentryALTinterwordstretchfactor}{4}
\providecommand{\BIBentryALTinterwordspacing}{\spaceskip=\fontdimen2\font plus
\BIBentryALTinterwordstretchfactor\fontdimen3\font minus
  \fontdimen4\font\relax}
\providecommand{\BIBforeignlanguage}[2]{{%
\expandafter\ifx\csname l@#1\endcsname\relax
\typeout{** WARNING: IEEEtran.bst: No hyphenation pattern has been}%
\typeout{** loaded for the language `#1'. Using the pattern for}%
\typeout{** the default language instead.}%
\else
\language=\csname l@#1\endcsname
\fi
#2}}
\providecommand{\BIBdecl}{\relax}
\BIBdecl

\bibitem{Massey85}
J.~L. Massey and P.~Mathys, ``The collision channel without feedback,''
  \emph{{IEEE} Trans. Inf. Theory}, vol.~31, no.~2, pp. 192--204, Mar. 1985.

\bibitem{Nguyen92}
N.~Q. A, L.~Gy{\"{o}}rfi, and J.~L. Massey, ``Constructions of binary
  constant-weight cyclic codes and cyclically permutable codes,'' \emph{{IEEE}
  Trans. Inf. Theory}, vol.~38, no.~3, pp. 940--949, May 1992.

\bibitem{GyorfiVajda93}
L.~Gy\"{o}fi and I.~Vajda, ``Construction of protocol sequences for
  multiple-access collision channel without feedback,'' \emph{{IEEE} Trans.
  Inf. Theory}, vol.~39, no.~5, pp. 1762--1765, Sep. 1993.

\bibitem{Wong07}
W.~S. Wong, ``New protocol sequences for random access channels without
  feedback,'' \emph{{IEEE} Trans. Inf. Theory}, vol.~53, no.~6, pp. 2060--2071,
  Jun. 2007.

\bibitem{CWS08}
C.~S. Chen, W.~S. Wong, and Y.-Q. Song, ``Constructions of robust protocol
  sequences for wireless sensor and ad hoc networks,'' \emph{{IEEE} Trans. Veh.
  Technol.}, vol.~57, no.~5, pp. 3053--3063, 2008.

\bibitem{SCSW}
K.~W. Shum, C.~S. Chen, C.~W. Sung, and W.~S. Wong, ``Shift-invariant protocol
  sequences for the collision channel without feedback,'' \emph{{IEEE} Trans.
  Inf. Theory}, vol.~55, no.~7, pp. 3312--3322, Jul. 2009.

\bibitem{CRT}
K.~W. Shum and W.~S. Wong, ``Construction and applications of {CRT}
  sequences,'' \emph{{IEEE} Trans. Inf. Theory}, vol.~56, no.~11, pp.
  5780--5795, Nov. 2010.

\bibitem{Ghez}
S.~Ghez, S.~Verd\'{u}, and S.~C. Schwartz, ``Stability properties of slotted
  {ALOHA} with multipacket reception capability,'' \emph{{IEEE} Trans. Autom.
  Control}, vol.~33, no.~7, pp. 640--649, 1988.

\bibitem{ZZL09}
Y.~Zhang, P.~Zheng, and S.~Liew, ``How does multiple-packet reception
  capability scale the performance of wireless local area networks?''
  \emph{{IEEE} Trans. Mobile Comput.}, vol.~8, no.~7, pp. 923--935, Jul. 2009.

\bibitem{Babich10}
F.~Babich and M.~Comisso, ``Theoretical analysis of asynchronous multi-packet
  reception in 802.11 networks,'' \emph{{IEEE} Trans. Commun.}, vol.~58, no.~6,
  pp. 1782--1794, Jun. 2010.

\bibitem{AlohaMPR13}
A.~I. B and T.~G. Venkatesh, ``Order statistics based analysis of pure {ALOHA}
  in channels with multipacket reception,'' \emph{{IEEE} Commun. Lett.},
  vol.~17, no.~10, pp. 2012--2015, Oct. 2013.

\bibitem{MPR_MAC13}
A.~Mukhopadhyay, N.~Mehta, and V.~Srinivasan, ``Design and analysis of an
  acknowledgment-aware asynchronous {MPR} {MAC} protocol for distributed
  {WLANs},'' \emph{{IEEE} Trans. Wireless Commun.}, vol.~12, no.~5, pp.
  2068--2079, May 2013.

\bibitem{Chan13}
D.~Chan, T.~Berger, and L.~Tong, ``Carrier sense multiple access communications
  on multipacket reception channels: Theory and applications to {IEEE} 802.11
  wireless networks,'' \emph{{IEEE} Trans. Commun.}, vol.~61, no.~1, pp.
  266--278, Jan. 2013.

\bibitem{Tinguely05}
S.~Tinguely, M.~Rezaeian, and A.~J. Grant, ``The collision channel with
  recovery,'' \emph{{IEEE} Trans. Inf. Theory}, vol.~51, no.~10, pp.
  3631--3638, Oct. 2005.

\bibitem{ZSW}
Y.~Zhang, K.~W. Shum, and W.~S. Wong, ``On pairwise shift-invariant protocol
  sequences,'' \emph{{IEEE} Commun. Lett.}, vol.~13, no.~6, pp. 453--455, Jun.
  2009.

\bibitem{Rocha00}
V.~C. da~Rocha, Jr., ``Protocol sequences for collision channel without
  feedback,'' \emph{IEE Electron. Lett.}, vol.~36, no.~24, pp. 2010--2012, Nov.
  2000.

\bibitem{Blind96}
S.~Talwar, M.~Viberg, and A.~Paulraj, ``Blind separation of synchronous
  co-channel digital signals using an antenna array, part {I}: Algorithms,''
  \emph{IEEE Trans. Signal Processing}, vol.~44, no.~5, pp. 1184--1197, May
  1996.

\bibitem{Blind98}
A.~van~der Veen, ``Algebraic methods for deterministic blind beamforming,''
  \emph{Proc. IEEE}, vol.~86, pp. 1987--2008, Oct. 1998.

\bibitem{Tong01}
L.~Tong, Q.~Zhao, and G.~Mergen, ``Multipacket reception in random access
  wireless networks: From signal processing to optimal medium access control,''
  \emph{IEEE Commun. Mag.}, pp. 108--112, Nov. 2001.

\end{thebibliography}
\end{document}